\documentclass{article}

\usepackage{url}
\usepackage{amssymb}
\usepackage{amsmath}
\usepackage{amsthm}
\usepackage{stmaryrd}

\newtheorem{lemma}{Lemma}
\newtheorem{theorem}[lemma]{Theorem}
\newtheorem{proposition}[lemma]{Proposition}
\newtheorem{corollary}[lemma]{Corollary}

\newtheorem*{convention}{Convention}
\theoremstyle{definition}

\theoremstyle{remark}
\newtheorem{example}[lemma]{Example}
\newtheorem{remark}[lemma]{Remark}

\newcommand{\AND}{\mathbin{\mathrm{AND}}}
\newcommand{\OPT}{\mathbin{\mathrm{OPT}}}
\newcommand{\bound}{\mathrm{bound}}
\newcommand{\UNION}{\mathbin{\mathrm{UNION}}}
\newcommand{\MINUS}{\mathbin{\mathrm{MINUS}}}

\newcommand{\FILTER}{\mathbin{\mathrm{FILTER}}}
\newcommand{\BIND}{\mathbin{\mathrm{BIND}}}
\newcommand{\VALUES}{{\textstyle\mathop{\mathrm{VALUES}}}}
\newcommand{\SELECT}{{\textstyle \mathop{\mathrm{SELECT}}}}
\newcommand{\EXISTS}{\mathop{\mathrm{EXISTS}}}

\newcommand{\dom}[1]{\mathrm{dom}(#1)}

\newcommand{\var}[1]{\mathrm{var}(#1)}
\newcommand{\varp}[2]{\mathrm{var}^{#1}(#2)}

\newcommand{\semm}[2]{\llbracket #1 \rrbracket_{#2}}
\newcommand{\sem}[1]{\semm{#1}G}

\newcommand{\adom}[1]{\mathrm{adom}(#1)}
\newcommand{\vadom}[1]{\mathit{adom}_{#1}}
\newcommand{\myfrag}[1]{\text{SPARQL}(#1)}
\newcommand{\frageq}{\myfrag{\bound,\allowbreak =,\allowbreak \neq_c}}
\newcommand{\fragneq}{\myfrag{\bound,\allowbreak \neq,\allowbreak \neq_c}}
\newcommand{\frageqc}{\myfrag{=_c}}
\newcommand{\fragnegbound}{\myfrag{\neg\bound}}
\newcommand{\frageqneq}{\myfrag{=,\allowbreak \neq}}

\begin{document}
\title{On the satisfiability problem \\ for SPARQL
patterns\thanks{This
work has been funded by grant G.0489.10 of
the Research Foundation Flanders (FWO).}}
\author{Xiaowang Zhang \\
\normalsize School of Computer Science and Technology \\
\normalsize Tianjin University\thanks{School
of Computer Science and Technology, Tianjin University, No.92 Weijin
Road, Nankai District, Tianjin 300072, P.R. China,
\protect \url{xiaowang.zhang@tju.edu.cn};
work performed while at Universiteit Hasselt.}
\and
Jan Van den Bussche \\ \normalsize Universiteit Hasselt\thanks{Databases and
Theoretical Computer Science, Universiteit Hasselt, Martelarenlaan
42, 3500 Hasselt, Belgium, \protect
\url{jan.vandenbussche@uhasselt.be}}
\and
Fran\c cois Picalausa \\ \normalsize \protect \url{fpicalausa@gmail.com}}

\maketitle

\begin{abstract}

The satisfiability problem for SPARQL patterns is undecidable in
general, since SPARQL~1.0 can express the relational algebra.
The goal of this paper is to delineate the boundary of
decidability of satisfiability in terms of the constraints
allowed in filter conditions.  The classes of constraints
considered are bound-constraints, negated bound-constraints,
equalities, nonequalities, constant-equalities, and
constant-nonequalities.  The main result of the paper can be
summarized by saying that, as soon as inconsistent filter
conditions can be formed, satisfiability is undecidable.  The key
insight in each case is to find a way to emulate the set
difference operation.  Undecidability can then be obtained from a
known undecidability result for the algebra of binary relations
with union, composition, and set difference.  When no
inconsistent filter conditions can be formed, satisfiability is
decidable by syntactic checks on bound variables and on the use
of literals.  Although the problem is shown to be NP-complete, it
is experimentally shown that the checks can be implemented
efficiently in practice.  The paper also points out that
satisfiability for the so-called `well-designed' patterns can be
decided by a check on bound variables and a check for
inconsistent filter conditions.

\end{abstract}

\section{Introduction}

The Resource Description Framework \cite{RDFprimer} is a
popular data model for information in the Web.  RDF represents
information in the form of directed, labeled graphs.  The
standard query language for RDF data is SPARQL \cite{sparql1.1}.
The current version~1.1 of SPARQL extends SPARQL~1.0
\cite{sparql} with important features such as aggregation and
regular path expressions \cite{chili_yotta}.
Other features, such as negation and
subqueries, have also been added, but mainly for efficiency
reasons, as they were already expressible, in a more involved manner, in version~1.0.  Hence, it is still relevant to study the
fundamental properties of SPARQL~1.0.  In this paper, we follow
the elegant formalization of SPARQL~1.0 by Arenas, Gutierrez and
P\'erez \cite{perez_sparql_tods,semanticsparql} which is
eminently suited for theoretical investigations.

The fundamental problem that we investigate is that of
\emph{satisfiability} of SPARQL patterns.  A pattern is called
satisfiable if there exists an RDF graph under which the pattern
evaluates to a nonempty set of mappings.  For any query language,
satisfiability is clearly one of the essential properties one
needs to understand if one wants to do automated reasoning.
Since SPARQL patterns can emulate relational algebra expressions
\cite{ag_expsparql,polleres_sparqldatalog,chile_sparql_pods}, and
satisfiability for relational algebra is undecidable
\cite{ahv_book}, the general satisfiability problem for SPARQL is
undecidable as well.

Whether or not a pattern is satisfiable depends mainly on the
filter operations appearing in the pattern.
The goal of this paper is
to precisely delineate the decidability of SPARQL fragments that
are defined in terms of the constraints that can be used as
filter conditions.  The six basic classes of constraints we
consider are bound-constraints; equalities; constant-equalities;
and their negations.  In this way, fragments of SPARQL can be
constructed by specifying which kinds of constraints are allowed
as filter conditions.  For example, in the fragment $\fragneq$,
filter conditions can only be bound constraints, nonequalities,
and constant-nonequalities.

Our main result states that the only fragments for which
satisfiability is decidable are the two fragments $\frageq$ and
$\fragneq$ and their subfragments.  Consequently, as soon as
either negated bound-constraints, or constant-equalities, or
combinations of equalities and nonequalities are allowed, the
satisfiability problem becomes undecidable.  Each undecidable
case is established by showing how the set difference operation
can be emulated.  This was already known using negated
bound-constraints \cite{ag_expsparql,chile_sparql_pods}; so we
show it is also possible using constant-equalities, and using
combinations of equalities and nonequalities, but in no other
way.  Undecidability can then be obtained from a known
undecidability result for the algebra of binary relations with
union, composition, and set difference \cite{tony_da_arxiv}.

In the decidable cases, satisfiability can be decided by
syntactic checks on bound variables and the use of literals.
Although the problem is shown to be NP-complete, it
is experimentally shown that the checks can be implemented
efficiently in practice.

At the end of the paper we look at a well-behaved class of
patterns known as the `well-designed' patterns
\cite{perez_sparql_tods}.  We observe that satisfiability of
well-designed patterns can be decided by combining the check on
bound variables with a check for inconsistent filter conditions.

This paper is further organized as follows. In the next section,
we introduce syntax and semantics of SPARQL patterns and
introduce the different fragments under consideration.
Section~\ref{secsat} introduces the satisfiability problem and
shows satisfiability checking for the fragments $\frageq$ and
$\fragneq$.  Section~\ref{secund} shows undecidability for the
fragments $\fragnegbound$, $\frageqc$, and $\frageqneq$.
Section~\ref{secwell} considers well-designed patterns.

Section~\ref{secexp} reports on experiments that test our
decision methods in practice.
In Section~\ref{sec1.1} we briefly discuss how our results extend to
the new operators that have been added to SPARQL~1.1.  We
conclude in Section~\ref{seconcl}.

\section{SPARQL and fragments}

In this section we recall the syntax and semantics of SPARQL
patterns, closely following the core SPARQL formalization given
by Arenas, Gutierrez and P\'erez
\cite{perez_sparql_tods,semanticsparql,chile_sparql_pods}.\footnote{The
cited works are seminal works on the semantics and complexity of
SPARQL patterns, but they do not investigate the satisfiability
of SPARQL patterns which is the main topic of the present paper.
The cited works also extensively discuss minor deviations between
the formalization and real SPARQL,
and why these differences are
inessential for the purpose of formal investigation.}
The semantics we use is set-based, whereas the semantics of real SPARQL
is bag-based.  However, for satisfiability (the main topic of
this paper), it makes no difference whether we use a set or bag
semantics \cite[Lemma~1]{schmidt_sparqloptim}.

In this section we will also define the language fragments defined in
terms of allowed filter conditions, which will form the object of
this paper.

\subsection{RDF graphs}

Let ${I}$, ${B}$, and ${L}$ be infinite sets of \emph{IRIs},
\emph{blank nodes} and \emph{literals}, respectively.  These
three sets are pairwise disjoint.  We denote the union $I \cup B
\cup L$ by $U$, and elements of $I \cup L$ will be referred to as
\emph{constants}.  Note that blank nodes are not constants.

A triple $(s, p, o) \in ({I}\cup {B}) \times {I} \times U$
is called an \emph{RDF triple}.  An \emph{RDF graph}
is a finite set of RDF triples.

\subsection{Syntax of SPARQL patterns}

Assume furthermore an infinite set $V$ of \emph{variables},
disjoint from $U$.  The convention in SPARQL is that variables
are written beginning with a question mark, to distinguish them
from constants. We will follow this convention in this paper.

SPARQL \emph{patterns} are inductively defined as follows.
\begin{itemize}
\item Any triple from $({I}\cup {L} \cup {V}) \times ({I} \cup
{V}) \times ({I} \cup {L} \cup {V}$) is a pattern (called a
\emph{triple pattern}).
\item If $P_{1}$ and $P_{2}$ are patterns, then so are
the following:
\begin{itemize}
\item
$P_{1} \UNION P_{2}$;
\item
$P_{1} \AND P_{2}$;
\item
$P_{1} \OPT P_{2}$.
\end{itemize}
\item
If $P$ is a pattern and $C$ is a
constraint (defined next),
then $P \FILTER C$ is a pattern; we call $C$ the
\emph{filter condition}.

Here, a \emph{constraint} can have one of the six following forms:
\begin{enumerate}
\item
\emph{bound-constraint:} $\bound(?x)$
\item
\emph{negated bound-constraint:} $\neg \bound(?x)$
\item
\emph{equality:} $?x={?y}$
\item
\emph{nonequality:} $?x\neq {?y}$ with $?x$ and $?y$ distinct
variables
\item
\emph{constant-equality:} $?x = c$ with $c$ a constant
\item
\emph{constant-nonequality:} $?x \neq c$
\end{enumerate}
\end{itemize}

We do not need to consider conjunctions and disjunctions in
filter conditions, since conjunctions can be expressed by
repeated application of filter, and disjunctions can be expressed
using UNION\@.  Hence, by going to disjunctive normal form, any
predicate built using negation, conjunction, and
disjunction is indirectly supported by our
language.

Moreover, real SPARQL also allows blank nodes in triple patterns.
This feature has been omitted from the formalization
\cite{perez_sparql_tods,semanticsparql,chile_sparql_pods},
because blank nodes in triple patterns can be equivalently
replaced by variables.

\subsection{Semantics of SPARQL patterns}

The semantics of patterns is defined in terms of sets of
so-called \emph{solution mappings}, hereinafter simply called
\emph{mappings}.  A solution mapping is a total function
$\mu : S \to U$ on some finite set $S$ of variables. 
We denote the domain $S$ of $\mu$ by $\dom \mu$.

We make use of the following convention.
\begin{convention} \label{covve}
For any mapping $\mu$ and any constant $c \in I \cup L$, we 
agree that $\mu(c)$ equals $c$ itself.  
\end{convention}
In other words, mappings are by default extended to constants according to the
identity mapping.

Now given a graph $G$ and a pattern $P$, we define the semantics
of $P$ on $G$, denoted by $\sem P$, as a set of mappings, in
the following manner.
\begin{itemize}
\item
If $P$ is a triple pattern $(u, v, w)$, then $$ \sem P :=
\{\mu : \{u, v, w\} \cap V \to U \mid (\mu(u),\mu(v),\mu(w)) \in G\}. $$
This definition relies on Convention~\ref{covve} formulated above.
\item
If $P$ is of the form $P_1 \UNION P_2$, then
$$ \sem P := \sem {P_1} \cup \sem {P_2}. $$
\item
If $P$ is of the form $P_1 \AND P_2$, then
$$ \sem P  := \sem {P_1}  \Join \sem {P_2} , $$
where, for any two sets of mappings $\Omega_1$ and $\Omega_2$,
we define
$$ \Omega_1 \Join \Omega_2 =
\{\mu_1 \cup \mu_2 \mid
\text{$\mu_1 \in \Omega_1$ and
$\mu_2 \in \Omega_2$ and $\mu_1 \sim \mu_2$} \}. $$
Here, two mappings $\mu_1$ and $\mu_2$ are called
\emph{compatible}, denoted by $\mu_1 \sim \mu_2$, if
they agree on the intersection of their domains, i.e.,
if for every variable $?x \in \dom {\mu_1} \cap \dom {\mu_2}$, we have
$\mu_1(?x) = \mu_2(?x)$.  Note that when $\mu_1$ and $\mu_2$ are
compatible, their union $\mu_1 \cup \mu_2$ is a well-defined
mapping; this property is used in the formal definition above.
\item
If $P$ is of the form $P_1 \OPT P_2$, then
$$ \sem P  := (\sem {P_1}  \Join \sem {P_2} )
\cup (\sem {P_1}  \smallsetminus \sem {P_2} ),
$$
where, for any two sets of mappings $\Omega_1$ and $\Omega_2$,
we define
$$ \Omega_1 \smallsetminus \Omega_2 =
\{ \mu_1 \in \Omega_1 \mid \neg \exists \mu_2 \in \Omega_2 :
\mu_1 \sim \mu_2\}. $$
\item
Finally, if $P$ is of the form $P_1 \FILTER C$, then
$$ \sem P  := \{\mu \in \sem {P_1}  \mid \mu \models C\} $$
where the satisfaction of a constraint $C$ by a mapping $\mu$,
denoted by $\mu \models C$, is defined as follows:
\begin{enumerate}
\item
$\mu \models \bound(?x)$ if $?x \in \dom \mu$;
\item
$\mu \models \neg \bound(?x)$ if $?x \notin \dom \mu$;
\item
$\mu \models {?x={?y}}$ if $?x,?y \in \dom \mu$ and $\mu(?x)=\mu(?y)$;
\item
$\mu \models {?x\neq {?y}}$ if $?x,?y \in \dom \mu$ and $\mu(?x)\neq\mu(?y)$;
\item
$\mu \models {?x=c}$ if $?x \in \dom \mu$ and $\mu(?x)=c$;
\item
$\mu \models {?x\neq c}$ if $?x \in \dom \mu$ and $\mu(?x)\neq c$.
\end{enumerate}
\end{itemize}

Note that $\mu \models {?x \neq {?y}}$ is not the same as $\mu
\not \models {?x = {?y}}$, and similarly for $\mu \models {?x
\neq c}$.  This is in line with the three-valued logic semantics
for filter conditions used in the official semantics
\cite{semanticsparql}.  For example, if $?x \notin \dom \mu$,
then in three-valued logic $?x=c$ evaluates to $\mathit{error}$
under $\mu$; consequently, also $\neg {?x=c}$ evaluates to
$\mathit{error}$ under $\mu$.  Accordingly, in the semantics
above, we have both $\mu \not \models {?x=c}$ and $\mu \not
\models {?x\neq c}$.

\subsection{SPARQL fragments}

We can form fragments of SPARQL by specifying which of the six
classes of constraints are allowed as filter conditions.  We
denote the class of bound-constraints by `bound', negated
bound-constraints by `$\neg\bound$', equalities by `$=$',
nonequalities by `$\neq$', constant-equalities by `$=_c$', and
constant-nonequalities by `$\neq_c$'.  Then for any subset $F$ of
$\{\bound,\neg \bound,=,\neq,=_c,\neq_c\}$ we can form the
fragment $\myfrag F$.  For example, in the fragment $\frageq$,
filter conditions can only be bound constraints, equalities, and
constant-nonequalities.

\section{Satisfiability: decidable fragments} \label{secsat}

A pattern $P$ is called \emph{satisfiable} if there exists a
graph $G$ such that $\sem P $ is nonempty.  In general,
checking satisfiability is a very complicated, indeed
undecidable, problem.  But for the two fragments $\frageq$ and
$\fragneq$, it will turn out that there are essentially only two
possible reasons for unsatisfiability.

The first possible reason is that the pattern specifies a literal
value in the first position of some RDF triple, whereas RDF
triples can only have literals in the third position.  For
example, using the literal 42, the triple pattern $(42,?x,?y)$ is
unsatisfiable.  Note that literals in the middle position of a triple
pattern are already disallowed by the definition of triple
pattern, so we only need to worry about the first position.

This discrepancy between triple patterns and RDF triples is easy
to sidestep, however.  In the Appendix we show how, without loss
of generality, we may assume from now on that \emph{patterns do
not contain any triple pattern $(u,v,w)$ where $u$ is a literal.}

The second and main possible reason for unsatisfiability is that filter
conditions require variables to be bound together in a way that
cannot be satisfied by the subpattern to which the filter
applies.  For example, the pattern $$ ((?x,a,?y) \UNION
(?x,b,?z)) \FILTER (\bound(?y) \land \bound(?z)) $$ is
unsatisfiable.  Note that bound constraints are not strictly
necessary to illustrate this phenomenon: if in the above example
we replace the filter condition by $?y={?z}$ the resulting pattern
is still unsatisfiable.

We next prove formally that satisfiability for patterns in $\frageq$ and
$\fragneq$ is effectively decidable, by catching the
reason for unsatisfiability described above.
Note also that the two fragments can not
be combined, since satisfiability for $\myfrag{=,\neq}$ is
undecidable as we will see in the next Section.

\subsection{Checking bound variables} \label{seccheck}

To perform bound checks on variables, we associate to every
pattern $P$ a set $\Gamma(P)$ of schemes,
where a \emph{scheme} is simply a set of variables, in the
following way.\footnote{We define $\Gamma(P)$ for general patterns,
not only for those belonging to the fragments considered in this
Section, because we will make another use of $\Gamma(P)$ in
Section~\ref{secwell}.}

\begin{itemize}
\item
If $P$ is a triple pattern $(u,v,w)$, then $\Gamma(P) :=
\{\{u,v,w\}\cap V\}$.
\item
$\Gamma(P_1 \UNION P_2) := \Gamma(P_1) \cup \Gamma(P_2)$.
\item
$\Gamma(P_1 \AND P_2) := \{S_1 \cup S_2 \mid S_1 \in \Gamma(P_1)$
and $S_2 \in \Gamma(P_2)\}$.
\item
$\Gamma(P_1 \OPT P_2) := \Gamma(P_1 \AND P_2) \cup \Gamma(P_1)$.
\item
$\Gamma(P_1 \FILTER C) := \{S \in \Gamma(P_1) \mid S \vdash 
C\}$, where $S \vdash C$ is defined as follows:
\begin{itemize}
\item
If $C$ is of the form $\bound(?x)$ or $?x=c$ or $?x \neq c$, then
$S \vdash C$ if $?x \in S$;
\item
If $C$ is of the form $?x={?y}$
or $?x\neq{?y}$, then
$S \vdash C$ if $?x,?y \in S$;
\item
$S \vdash {\neg \bound(?x)}$ if $?x \notin S$.
\end{itemize}
\end{itemize}

\begin{example} \label{exgammaunion}
Consider the pattern
$$ P = (?x,p,?y) \OPT ( (?x,q,?z) \UNION (?x,r,?u)). $$
For the subpattern $P_1 = 
(?x,q,?z) \UNION (?x,r,?u)$ we have $\Gamma(P_1) =
\{\{?x,?z\},\allowbreak \{?x,?u\}\}$.  Hence, $\Gamma((?x,p,?y) \AND P_1) =
\{\{?x,?y,?z\},\{?x,?y,?u\}\}$. We conclude that $\Gamma(P) =
\{\{?x,?y\},\{?x,?y,?z\},\{?x,?y,?u\}\}$.
\end{example}

\begin{example} \label{exgamma}
For another example,
consider the pattern
$$ P = ((?x, p, ?y) \OPT ((?x, q, ?z)
\FILTER {?y = ?z})) \FILTER {?x \neq c}. $$
We have $\Gamma(?x,q,?z)
= \{\{?x,?z\}\}$.  Note that $\{?x,?z\} \not \vdash {?y = ?z}$,
because $?y \notin \{?x,?z\}$.  Hence, for the subpattern $P_1 =
(?x, q, ?z) \FILTER {?y = ?z}$ we have $\Gamma(P_1)=\emptyset$.
For the subpattern $P_2 = (?x,p,?y) \OPT P_1$ we then have
$\Gamma(P_2)=\Gamma(?x,p,?y)=\{\{?x,?y\}\}$.  Since $\{?x,?y\}
\vdash {?x \neq c}$, we conclude that $\Gamma(p) = \{\{?x,?y\}\}$.
\qed
\end{example}

We now establish the main result of this Section.

\begin{theorem} \label{theordecidable}
Let $P$ be a $\frageq$ or $\fragneq$ pattern.  Then $P$ is
satisfiable if and only if\/ $\Gamma(P)$ is nonempty.
\end{theorem}

The only-if direction of Theorem~\ref{theordecidable} is the easy
direction and is given by the following Lemma~\ref{lemoif}.
Note that this lemma holds for general patterns; it can be
straightforwardly proven by induction on the structure of $P$.

\begin{lemma} \label{lemoif}
Let $P$ be a pattern.  If $\mu \in \sem P $ then there exists $S
\in \Gamma(P)$ such that $\dom \mu = S$.
\end{lemma}

The if direction of Theorem~\ref{theordecidable} for $\frageq$ is given by
the following Lemma~\ref{lemifeq}.

In the following
we use $\var P$ to denote the set of all variables occurring in
a pattern $P$.\footnote{We
also use the following standard notion of restriction of
a mapping. If $f : X \to Y$ is a total function and $Z \subseteq X$,
then the restriction $f|_Z$ of $f$ to $Z$ is the total function from $Z$
to $Y$ defined by $f|_Z(z)=f(z)$ for every $z \in Z$.  That is,
$f|_Z$ is the same as $f$ but is only defined on the subdomain $Z$.}

\begin{lemma} \label{lemifeq}
Let $P$ be a pattern in $\frageq$.
Let $c \in I$ be a constant that does not appear in any
constant-nonequality filter condition in $P$.
With the constant mapping $\mu : \var P \to \{c\}$,
let $G$ be the RDF graph consisting of all possible triples
$(\mu(u),\mu(v),\mu(w))$ where $(u,v,w)$ is a triple
pattern in $P$.

Then for every $S \in \Gamma(P)$ there exists $S' \supseteq S$
such that $\mu|_{S'}$ belongs to $\sem P$.
\end{lemma}
\begin{proof}
By induction on the structure of $P$.
If $P$ is a triple pattern $(u,v,w)$ then $S = \{u,v,w\}\cap V$.  Since
$(\mu|_S(u),\mu|_S(v),\mu|_S(w))=(\mu(u),\mu(v),\mu(w)) \in G$, we have
$\mu|_S \in \sem P $ and we can take $S'=S$.

If $P$ is of the form $P_1 \UNION P_2$, then the claim follows readily
by induction.

If $P$ is of the form $P_1 \AND P_2$, then we have $S=S_1 \cup S_2$
with $S_i \in \Gamma(P_i)$ for $i=1,2$.  By
induction, there exists $S_i' \supseteq S_i$ such that
$\mu|_{S_i'} \in \sem {P_i} $.
Clearly $\mu|_{S_1'} \sim
\mu|_{S_2'}$ since they are restrictions of the same mapping.
Hence $\mu|_{S_1'} \cup \mu|_{S_2'} = \mu_{S_1' \cup S_2'} \in
\sem P$ and we can take $S' = S_1' \cup S_2'$.

If $P$ is of the form $P_1 \OPT P_2$, then there are two possibilities.
\begin{itemize}
\item
If $S \in \Gamma(P_1 \AND P_2)$ then we can reason as in the previous
case.
\item
If $S \in \Gamma(P_1)$ then by induction there exists $S'_1
\supseteq S$ so that $\mu|_{S'_1} \in \sem {P_1}$.
Now there are two further possibilities:
\begin{itemize}
\item
If $\Gamma(P_2)$ is nonempty then by induction there exists some
$S_2'$ so that $\mu|_{S_2'} \in \sem {P_2}$.
We can now reason again as in the case $P_1 \AND P_2$.
\item
Otherwise, by Lemma~\ref{lemoif} we know that $\sem {P_2} $ is
empty.  But then $\sem P = \sem {P_1} $ and we can take $S'
= S'_1$.
\end{itemize}
\end{itemize}

Finally, if $P$ is of the form $P_1 \FILTER C$, then we know that $S
\in \Gamma(P_1)$ and $S \vdash C$.  By induction, there
exists $S' \supseteq S$ such that $\mu|_{S'} \in
\sem {P_1} $.  We show that $\mu|_{S'} \in
\sem P$ by showing that $\mu|_{S'} \models C$.
There are three possibilities for $C$.
\begin{itemize}
\item
If $C$ is of the form $\bound(?x)$, then we know by $S \vdash C$
that $?x \in S'$.  Hence $\mu|_{S'} \models C$.
\item
If $C$ is of the form $?x = {?y}$, then we again know $?x,?y \in S'$,
and certainly $\mu|_{S'} \models C$ since $\mu$ maps everything
to $c$.
\item
If $C$ is of the form $?x \neq d$, then we have
$d \neq c$ by the choice of $c$, so
$\mu|_{S'} \models C$ since $\mu(?x)=c$.
\qedhere
\end{itemize}
\end{proof}

\begin{example}
To illustrate the above Lemma, consider the pattern
$$ P = ((?x,p,?y) \FILTER {?x \neq a}) \OPT ( (?x,q,?z) \UNION (?x,r,?u) ) $$
which is a variant of the pattern
from Example~\ref{exgammaunion}.  As in that example, we have
$\Gamma(P) = \{\{?x,?y\},\{?x,?y,?z\},\{?x,?y,?u\}\}$.
In this case,
the mapping $\mu$ from the Lemma maps $?x$, $?y$, $?z$ and $?u$ to $c$.
The graph $G$ from the Lemma equals
$\{(c,p,c),(c,q,c),(c,r,c)\}$, and $\sem P = \{\mu_1,\mu_2\}$
where $\mu_1 = \mu|_{\{?x,?y,?z\}}$ and
$\mu_2 = \mu|_{\{?x,?y,?u\}}$.  Now consider $S = \{?x,?y\} \in
\Gamma(P)$.  Then for $S'=\{?x,?y,?z\}$ we indeed have $S'
\supseteq S$ and $\mu|_{S'} = \mu_1 \in \sem P$.  Note that in
this example we could also have chosen $\{?x,?y,?u\}$ for $S'$.
\qed
\end{example}

The counterpart to Lemma~\ref{lemifeq} for the fragment
$\fragneq$ is given by the following Lemma, thus settling
Theorem~\ref{theordecidable} for that fragment.

\begin{lemma} \label{lemifneq}
Let $P$ be a pattern in $\fragneq$.
Let $W$ be the set of all constants appearing in a
constant-nonequality filter condition in $P$.  Let $Z \subseteq
I$ be a finite set of constants of the same cardinality as $\var P$, 
and disjoint from $W$.  With $\mu : \var P \to Z$ an arbitrary
but fixed injective mapping,
let $G$ be the RDF graph consisting of
all possible triples $(\mu(u),\mu(v),\mu(w))$
where $(u,v,w)$ is a triple pattern in $P$.

Then for every $S \in \Gamma(P)$
there exists $S' \supseteq S$ such that
$\mu|_{S'}$ belongs to $\sem P $.
\end{lemma}
\begin{proof}
We prove for every subpattern $Q$ of $P$ that for every $S \in
\Gamma(Q)$ there exists $S' \supseteq S$ such that $\mu|_{S'}
\in \sem Q $.  The proof is by induction on the height of $Q$.
The reasoning is largely the same as in the proof of
Lemma~\ref{lemifeq}.  The only difference is
in the case where $Q$ is of the form
$Q_1 \FILTER C$.  In showing that $\mu_{S'} \models C$, we
now argue as follows for the last two cases:

\begin{itemize}
\item
If $C$ is of the form $?x \neq {?y}$, then $\mu|_{S'} \models C$ since $\mu$ is injective.
\item
If $C$ is of the form $?x \neq c$, then $\mu|_{S'} \models C$ since $Z$ and $W$ are disjoint.
\qedhere
\end{itemize}
\end{proof}

\subsection{Computational complexity}

In this section we show that satisfiability for the decidable
fragments is NP-complete.  Note that this does not immediately
follow from the NP-completeness of SAT, since boolean formulas
are not part of the syntax of the decidable fragments.

Theorem~\ref{theordecidable} implies the following
complexity upper bound:

\begin{corollary}
The satisfiability problem for
$\frageq$ patterns, as well as for $\fragneq$ patterns,
belongs to the complexity class
NP.
\end{corollary}
\begin{proof}
By Theorem~\ref{theordecidable}, a
$\frageq$ or $\fragneq$ pattern $P$ is satisfiable if and only if
there exists a scheme in $\Gamma(P)$.  Following the
definition of $\Gamma(P)$, it is clear that there is a
polynomial-time nondeterministic algorithm such that, on
input $P$, each accepting possible run computes a scheme in
$\Gamma(P)$, and such that every scheme in $\Gamma(P)$ is
computed by some accepting possible run.

Specifically,
the algorithm works bottom-up on the syntax tree of $P$ and
computes a scheme for every subpattern.  At every leaf $Q$,
corresponding to a triple pattern in $P$, we compute the unique
scheme in $\Gamma(Q)$.  At every UNION operator we
nondeterministically choose between continuing with the scheme
from the left or from right child.  At every AND operator we
continue with the union of the left and right child schemes.  At
every OPT operator, we nondeterministically choose between
treating it as an AND, or simply continuing with the scheme from
the left.  At every FILTER operation with constraint $C$ we check
for the child scheme $S$ whether $S \vdash C$.  If the check
succeeds, we continue with $S$; if the check fails, the run is
rejected.  When the computation has reached the root of the
syntax tree and we can compute a scheme for the root, the run is
accepting and the computed scheme is the output.
\end{proof}

We next show that satisfiability is actually NP-hard, even for
patterns not using any OPT operators and using only bound
constraints in filter conditions.

\begin{proposition} \label{prophard}
The satisfiability problem for OPT-free patterns in
the fragment $\myfrag{\bound}$ is NP-hard.
\end{proposition}
\begin{proof}
We define the problem Nested Set Cover as follows:
\begin{description}
\item[Input:]  A finite set $T$ and a finite set $E$ of sets of
subsets of $T$.  (So, every element of $E$ is a set of subsets of
$T$.)
\item[Decide:] Whether for each element $e$ of $E$ we can choose
a subset $S_e$ in $e$, so that $\bigcup_{e \in E} S_e = T$.
\end{description}

Let us first describe how the above problem can be reduced in
polynomial time to the satisfiability problem at hand.
Consider an input $(T,E)$ for Nested Set Cover.  Without
loss of generality we may assume that $T$ is a set of variables
$\{?x_1,?x_2,\dots,?x_n\}$.  Fix some
constant $c$.  For
any subset $S$ of $T$, we can make a pattern $P_S$ by taking the
AND of all $(x,c,c)$ for $x \in S$.  Now for a set $e$ of subsets
of $T$, we can form the pattern $P_e$ by taking the UNION of all
$P_S$ for $S \in e$.  Finally, we form the pattern $P_E$ by
taking the AND of all $P_e$ for $e \in E$.

Now consider the following pattern which we denote by $P_{(T,E)}$:
$$ P_E
\FILTER {\bound(?x_1)}
\FILTER {\bound(?x_2)}
\ldots
\FILTER {\bound(?x_n)} $$

We claim that $P_{(T,E)}$
is satisfiable if and only if $(T,E)$
is a yes-instance for Nested Set Cover.  To see the only-if
direction, let $G$ be a graph such that $\sem{P_{(T,E)}}$ is
nonempty, i.e., has as an element some solution mapping $\mu$.
Then in particular $\mu \in \sem{P_E}$.  Hence,
for every $e \in E$ there exists $\mu_e \in \sem{P_e}$ such that
$\mu = \bigcup_{e \in E} \mu_e$.  Since $P_e$ is the UNION of all
$P_S$ for $S \in e$, for each $e \in E$ there exists $S_e \in e$
such that $\mu_e \in \sem{P_{S_e}}$.  Since $P_{S_e}$ is the AND
of all $(x,c,c)$ for $x \in S_e$, it follows that
$\dom{\mu_e}=S_e$.  Hence, since $\dom \mu = \bigcup_{e\in E}
\dom{\mu_e}$, we have $\dom \mu = \bigcup_{e \in E} S_e$.
However, by the bound constraints in the filters applied in
$P_{(T,E)}$, we also have $\dom \mu = \{?x_1,\dots,?x_n\} = T$.
We conclude that $T = \bigcup_{e \in E} S_e$ as desired.

For the if-direction, assume that for each $e \in E$ there exists
$S_e \in e$ such that $T = \bigcup_{e \in E} S_e$.  Consider
the singleton graph $G = \{(c,c,c)\}$.  For any subset $S$ of
$T$, let $\mu_S : S \to \{c\}$ be the constant solution mapping
with domain $S$.  Clearly, $\mu_{S} \in \sem{P_S}$, so $\mu_{S_e}
\in \sem{P_e}$ for every $e \in E$.  All the $\mu_S$ map to the
same constant, so they are all compatible.  Hence, for $\mu =
\bigcup_{e \in E} \mu_{S_e}$, we have $\mu \in \sem{P_E}$.  Since
$\dom \mu = \bigcup_{e\in E} \dom{\mu_{S_e}} = \bigcup_{e \in E}
S_e = T = \{?x_1,\dots,?x_n\}$, the mapping $\mu$ satisfies every
constraint $\bound(?x_i)$ for $i=1,\dots,n$.  We conclude that
$\mu \in \sem{P_{(E,T)}}$ as desired.

It remains to show that Nested Set Cover is NP-hard.
Thereto we
reduce the classical CNF-SAT problem.  Assume given a boolean
formula $\phi$ in CNF,
so $\phi$ is a conjunction of clauses, where each clauses is a
disjunction of literals (variables or negated variables).  We
construct an input $(T,E)$ for Nested Set Cover as follows.
Denote the set of variables used in $\phi$ by $W$.

For $T$ we take the set of clauses of $\phi$.
For any variable $x \in W$, consider the set ${\rm Pos}_x$ consisting of all
clauses that contain a positive occurrence of $x$, and the set
${\rm Neg}_x$ consisting of all clauses that contain a negative
occurrence of $x$.  Then we define $e_x$ as the pair $\{{\rm
Pos}_x,{\rm Neg}_x\}$.

Now $E$ is defined as the set $\{e_x \mid x \in W\}$.
It is clear that $\phi$ is satisfiable if and only if the
constructed input is a yes-instance for Nested Set Cover.  Indeed,
truth assignments to the variables correspond to selecting either
${\rm Pos}_x$ or ${\rm Neg}_x$ from $e_x$ for each $x \in W$.
\end{proof}

\section{Undecidable fragments} \label{secund}

In this Section we show that the two decidable fragments $\frageq$ and
$\fragneq$ are, in a sense, maximal.  Specifically, the
three minimal fragments not subsumed by one of these two
fragments are $\fragnegbound$, $\frageqneq$, and $\frageqc$.  The
main result of this Section is:

\begin{theorem}
Satisfiability is undecidable for $\fragnegbound$ patterns,
for $\frageqneq$ patterns, and for $\frageqc$ patterns.
\end{theorem}

We will first present
the proof for $\fragnegbound$; after that we explain how the proof
can be adapted for the other two fragments.

\subsection{$\fragnegbound$}

Our approach is to reduce
from the satisfiability problem for the algebra of finite binary
relations with union, difference, and composition
\cite{tony_da_arxiv}.  This algebra is also called the Downward
Algebra and denoted by DA\@.  The expressions of DA are defined
as follows.  Let $R$ be an arbitrary fixed binary relation symbol.  
\begin{itemize}
\item
The symbol $R$ is a DA-expression.
\item
If $e_1$ and $e_2$ are DA-expressions, then so are $e_1 \cup
e_2$, $e_1 - e_2$, and $e_1 \circ e_2$.
\end{itemize}

Semantically, DA-expressions represent binary queries on binary
relations, i.e., mappings from binary relations to binary
relations.  Let $J$ be a binary relation.  For DA-expression $e$,  
we define the binary relation $e(J)$ inductively as follows:
\begin{itemize}
\item
$R(J) = J$;
\item
$(e_1 \cup e_2)(J) = e_1(J) \cup e_2(J)$;
\item
$(e_1 - e_2)(J) = e_1(J) - e_2(J)$ (set difference);
\item
$(e_1 \circ e_2)(J) = \{(x,z) \mid \exists y : (x,y) \in e_1(J)$
and $(y,z) \in e_2(J)\}$.
\end{itemize}

A DA-expression is called \emph{satisfiable} if there exists a
finite binary relation $J$ such that $e(J)$ is nonempty.

\begin{example}
An example of a DA-expression is
$e=(R \circ R)-R$.  If $J$ is the binary relation
$\{(a,b),(b,c),(a,c),(c,d)\}$ then $e(J) =
\{(b,d),(a,d)\}$.
An example of an unsatisfiable DA
expression is $(R \circ R - R) \circ R - R \circ R \circ R$.
\qed
\end{example}

We recall the following result.  It is actually well known
\cite{andreka_memoir} that
relational composition together with union and complementation
leads to an undecidable algebra; the following result simplifies
matters by showing that undecidability already holds for expressions
over a single relation symbol and using set difference instead of
complementation.  The following result has been proven
by reduction from the universality problem for context-free grammars.
\begin{theorem}[\cite{tony_da_arxiv}]
The satisfiability problem for DA-expressions is undecidable.
\end{theorem}

We are now ready to formulate the reduction from the
satisfiability problem for DA to the satisfiability problem for
$\fragnegbound$.

\begin{lemma} \label{reductionlemma}
Let $r \in I$ be an arbitrary fixed constant.
For any binary relation $J$, let $G_J$ be the RDF graph
$\{(c,r,d) \mid (c,d) \in J\}$.
Then for every DA-expression $e$ there exists a $\fragnegbound$ pattern
$P_e$ with the following properties:
\begin{enumerate}
\item
there exist two distinct fixed variables $?x$ and $?y$ such that for every RDF graph $G$ and every $\mu \in \sem {P_e}$, $?x$ and $?y$ belong to $\dom \mu$;
\item
for every binary relation $J$, we have $$ e(J) =
\{(\mu(?x),\mu(?y)) \mid \mu \in \semm {P_e}{G_J}\}; $$
\item
for every RDF graph $G$, we have $\sem {P_e} = \semm {P_e}{G^r}$,
where $G^r := \{(u,v,w) \in G \mid v=r\}$.
\end{enumerate}
\end{lemma}
\begin{proof}
By induction on the structure of $e$.  If $e$ is $R$ then $P_e$
is the triple pattern $(?x,r,?y)$.

If $e$ is of the form $e_1 \cup e_2$, then $P_e$ is $P_{e_1}
\UNION P_{e_2}$.

If $e$ is of the form $e_1 \circ e_2$, then $P_e$ is $P'_{e_1}
\AND P'_{e_2}$, where $P'_{e_1}$ and $P'_{e_2}$ are obtained as
follows.  First, by renaming variables, we may assume without
loss of generality that $P_{e_1}$ and $P_{e_2}$ have no variables
in common other than $?x$ and $?y$.  Let $?z$ be a fresh
variable.  Now in $P_{e_1}$, rename $?y$ to $?z$, yielding
$P'_{e_1}$, and in $P_{e_2}$, rename $?x$ to $?z$, yielding
$P'_{e_2}$.

Finally, if $e$ is of the form $e_1 - e_2$, then we use a known
idea \cite{chile_sparql_pods}.  As before we may assume
without loss of generality that $P_{e_1}$ and $P_{e_2}$ have no
variables in common other than $?x$ and $?y$.  Let $?u$ and $?w$
be two fresh variables.  Then $P_e$ is equal to $$ \bigl ( P_{e_1}
\OPT (P_{e_2} \AND (?u,r,?w)) \bigr ) \FILTER {\neg \bound(?u)}.
$$
\end{proof}

The above lemma provides us with a reduction from
satisfiability for DA to satisfiability for $\fragnegbound$, thus
showing undecidability of the latter problem.
Indeed, if $e$ is satisfiable, then clearly $P_e$ is satisfiable
as well, by property~2 of the lemma.  Conversely, if $P_e$ is
satisfiable by some RDF graph $G$, then, by
property~3 of the lemma, $\semm {P_e}{G^r}$ is nonempty.  Now
define the binary relation $J = \{(c,d) \mid (c,r,d) \in G\}$.
Then $G_J = G^r$, so by property~2 of the lemma we obtain the
nonemptiness of $e(J)$ as desired.

\subsection{$\frageqneq$} \label{seceqneq}

We now consider a minor variant of satisfiability for DA-expressions
where we restrict attention to binary relations over at least two
elements.  Formally, 
the \emph{active domain} of a binary relation $J$ is the set of all
entries in pairs belonging to $J$, so $\adom J := \{x \mid
\exists y : (x,y) \in J$ or $(y,x) \in J\}$.
Then a DA-expression $e$ is called \emph{two-satisfiable} if 
$e(J)$ is nonempty for some $J$ such that $\adom J$ has at least
two distinct elements.

Clearly, two-satisfiability is undecidable as well, for if it
were decidable, then satisfiability would be decidable too.
Indeed, $e$ is satisfiable if and only if it is two-satisfiable,
or satisfiable by a binary relation $J$
over a single element.  Up to isomorphism there is only one such
$J$ (the singleton $\{(x,x)\}$), and DA-expressions commute with
isomorphisms.

Lemma~\ref{reductionlemma} can now be adapted as follows.
Property~2 of the lemma is only claimed for every
binary relations $J$ over at least two distinct
elements.  In the proof for the case where $e$ is $e_1 - e_2$, we
use six fresh variables $?u$, $?u'$, $?v$, $?v'$, $?w$, and $?w'$.  We use the
abbreviation $\vadom{?u}$ for $(?u,r,?w) \UNION (?v,r,?u)$ and
similarly for $\vadom{?u'}$.  We 
now use the following pattern for $P_e$:
\begin{multline*}
\Bigl ( \bigl ( P_{e_1} \OPT
(
(P_{e_2} \AND \vadom{?u} \AND \vadom{?u'}) \FILTER {?u \neq {?u'}}
) \bigr) \\
{} \AND \vadom{?u} \AND \vadom{?u'} \Bigr ) \FILTER {?u={?u'}}. 
\end{multline*}

Let us verify that $P_e$ satisfies the three properties of
Lemma~\ref{reductionlemma}.
\begin{proof}
\begin{enumerate}
\item
By induction, $P_{e_1}$ has the property that
every returned solution mapping
has $?x$ and $?y$ in its domain.  Since $P_e$ is of the form
$$(P_{e_1} \OPT \ldots) \FILTER \ldots$$ the same property holds for $P_e$.

\item
Let $J$ be a binary relation on at least two distinct elements.
To prove the equality $$ e(J) = \{(\mu(?x),\mu(?y)) \mid \mu \in \semm
{P_e}{G_J}\} $$ we are going to consider both inclusions.  For easy reference we
name some subpatterns of $P_e$ as follows.
\begin{itemize}
\item
$P_2$ denotes
$(P_{e_2} \AND \vadom{?u} \AND \vadom{?u'}) \FILTER {?u \neq {?u'}}$;
\item
$P_3$ denotes $P_{e_1} \OPT P_2$.
\item
Thus, $P$ is
$(P_3 \AND \vadom{?u} \AND \vadom{?u'}) \FILTER {?u = {?u'}}$.
\end{itemize}

To prove the inclusion from right to left, let $\mu \in \semm
{P_e}{G_J}$.  Then $\mu = \mu_3 \cup \varepsilon$, where $\mu_3 \in
\semm {P_3}{G_J}$ and $\varepsilon$ is a mapping defined on $?u,$ and $?u'$
such that $\varepsilon(?u) = \varepsilon(?u')$.  In particular, $\mu_3 \sim
\varepsilon$.  Since $P_3 = P_{e_1} \OPT
P_2$, there are two possibilities for $\mu_3$:
\begin{itemize}
\item
$\mu_3 \in \semm {P_{e_1}}{G_J}$ and there is no $\mu_2 \in \semm
{P_2}{G_J}$ such that $\mu_3 \sim \mu_2$.  By induction, both
$?x$ and $?y$ belong to $\dom {\mu_3}$, so $(\mu(?x),\mu(?y))$
equals $(\mu_3(?x),\mu_3(?y))$, which belongs to $e_1(J)$ again by
induction.  So it remains to show that $(\mu(?x),\mu(?y)) \notin
e_2(J)$.  Assume the contrary.
Then there exists $\mu'_2 \in \semm {P_{e_2}}{G_J}$ 
such that $(\mu_3(?x),\mu_3(?y)) = (\mu'_2(?x),\mu'_2(?y))$.  
Since $\adom J$ has at least two distinct elements,
$\mu_2'$ can be extended to a mapping $\mu_2 \in \semm {P_2}{G_J}$.
Since $?x$ and $?y$ are the only
variables common to $\var {P_{e_1}}$ and $\var {P_2}$, we conclude
$\mu_3 \sim \mu_2$ which is a contradiction.
\item
$\mu_3 = \mu_1 \cup \mu_2$ with $\mu_1 \in \semm {P_{e_1}}{G_J}$
and $\mu_2 \in \semm {P_2}{G_J}$.  In particular, $\mu_3$ is
defined on $?u$ and $?u'$ and $\mu_3(?u) \neq \mu_3(?u')$.
On the other hand, since $\mu_3 \sim \varepsilon$, and
$\varepsilon(?u) =
\varepsilon(?u')$, also $\mu_3(?u) = \mu_3(?u')$.  This is a
contradiction, so the possibility under consideration cannot
happen.
\end{itemize}

To prove the inclusion from left to right, let $(c,d) \in e(J)$.
Since $(c,d) \in e_1(J)$, there exists $\mu_1 \in \semm
{P_{e_1}}{G_J}$ such that $(c,d) = (\mu_1(?x),\mu_1(?y))$.  
Assume, for the sake of argument, that there \emph{would} exist $\mu_2
\in \semm {P_2}{G_J}$ such that $\mu_1 \sim \mu_2$.  Mapping
$\mu_2$ contains a mapping $\mu_2' \in \semm {P_{e_2}}{G_J}$, by 
definition of $P_2$.  Since $(\mu_2'(?x),\mu_2'(?y)) \in e_2(J)$
and $\mu_1 \sim \mu_2$, it follows that $(c,d) \in e_2(J)$ which
is a contradiction.

So, we now know that there does \emph{not} exist
$\mu_2 \in \semm {P_2}{G_J}$ such that $\mu_1 \sim \mu_2$.
Hence, $\mu_1 \in \semm {P_3}{G_J}$.  Note that the
six variables $?u$, $?u'$, $?v$, $?v'$, $?w$, and $?w'$ 
do not belong to $\dom {\mu_1}$.  Since $J$ is nonempty,
$\mu_1$ can thus be extended to a mapping $\mu \in
\semm P{G_J}$.  We conclude $(c,d) = (\mu_1(?x),\mu_1(?y)) =
(\mu(?x),\mu(?y))$ as desired.

\item
The third property of
Lemma~\ref{reductionlemma} holds because
$\semm {\vadom{?u}}{G} = \semm {\vadom{?u}}{G^r}$ (and similarly
for $\vadom{?u'}$).
\end{enumerate}
\end{proof}

Using the adapted lemma, we can now reduce two-satisfiability for
DA to satisfiability for $\frageqneq$.  Indeed, a DA-expression
$e$ is two-satisfiable if and only if the pattern $$ P_e \AND
((\vadom{?u} \AND \vadom{?u'}) \FILTER {?u \neq {?u'}}) $$ is
satisfiable, where all variables used in $\vadom{?u}$ and
$\vadom{?u'}$ are distinct and disjoint from those used in $P_e$.

\subsection{$\frageqc$} \label{secfrageqc}

We consider a further variant of two-satisfiability, called
\emph{$ab$-satisfiability}, for two arbitrary fixed constants
$a,b \in I$ that are distinct from the constant $r$ already used
for Lemma~\ref{reductionlemma}.  A DA-expression is called
$ab$-satisfiable if $e(J)$ is nonempty for some binary relation
$J$ where $a,b \in \adom J$.

Since DA-expressions do not distinguish between isomorphic binary
relations, $ab$-satisfiability is equivalent to
two-satisfiability, and thus still undecidable.

We now again adapt Lemma~\ref{reductionlemma}, as follows.  
Property~2 is only claimed for every
binary relation $J$ such that $a,b \in
\adom J$.  In the proof for the case $e=e_1- e_2$, we now use the
following pattern for $P_e$:
\[
\Bigl ( \bigl ( P_{e_1} \OPT
((P_{e_2} \AND \vadom{?u}) \FILTER {?u = a})\bigr )
\AND \vadom{?u} \Bigr ) \FILTER {?u=b}. 
\]

The proof correctness of this construction is analogous to the
proof given in the previous Section~\ref{seceqneq}; instead of
exploiting the inconsistency between $?u \neq ?u'$ and $?u = ?u'$
as done in that proof, we now exploit the inconsistency between
$?u=a$ and $?u=b$.

We then obtain that $e$ is $ab$-satisfiable if and only if $$ P_e
\AND (\vadom{?u} \AND \vadom{?u'}) \FILTER {?u=a} \FILTER {?u'=b} $$
is satisfiable, establishing a reduction from $ab$-satisfiability
for DA to satisfiability for $\frageqc$.

\begin{remark}
Recall that literals cannot appear in first or second position in
an RDF triple.  Patterns using constant-equality predicates can
be unsatisfiable because of that reason.  For example, using the
literal 42, the pattern $(?x,?y,?z) \FILTER {?y=42}$ is
unsatisfiable.  However, we have seen here that
the use of constant-equality predicates leads to undecidability
of satisfiability for a much more fundamental reason, that has
nothing to do with literals, namely, the
ability to emulate set difference.
\end{remark}

\section{Satisfiability of well-designed patterns}
\label{secwell}

The \emph{well-designed} patterns \cite{perez_sparql_tods} have
been identified as a well-behaved class of SPARQL patterns, with
properties similar to the conjunctive queries for relational
databases \cite{ahv_book}.  Standard conjunctive queries are
always satisfiable, and conjunctive queries extended with
equality and nonequality constraints, possibly involving
constants, can only be unsatisfiable if the constraints are
inconsistent.  An analogous behavior is present in what we call
\emph{AF-patterns}: patterns that only use the AND and FILTER
operators. We will formalize this in Proposition~\ref{propaf}.
We will then show in Theorem~\ref{theorwell} that a well-designed
pattern is satisfiable if and only if its reduction to an
AF-pattern is satisfiable.  In other words, as far as
satisfiability is concerned, well-designed patterns can be
treated like AF-patterns.

\subsection{Satisfiability of AF-patterns}

In Section~\ref{seccheck} we have associated a set of schemes
$\Gamma(P)$ to every pattern $P$.  When $\Gamma(P)$ is empty, $P$
is unsatisfiable (Lemma~\ref{lemoif}).

Now when $P$ is an AF-pattern and $\Gamma(P)$ is nonempty,
the satisfiability of $P$ will turn out to depend solely on the
equalities, nonequalities, constant-equalities, and
constant-nonequalities occurring as filter conditions in $P$.  
We will denote the set of these constraints by $C(P)$.

Any set $\Sigma$ of constraints is called
\emph{consistent} if there exists a mapping
that satisfies every constraint in $\Sigma$.

We establish:

\begin{proposition} \label{propaf}
An AF-pattern $P$ is satisfiable if and only if\/ $\Gamma(P)$ is
non-empty and $C(P)$ is consistent.
\end{proposition}
\begin{proof}
The only-if direction of this proposition is given by
Lemma~\ref{lemoif} together with the
observation that if $\mu \in \sem P$, then $\mu$ satisfies
every constraint in $C(P)$.  Since $P$ is satisfiable, such $G$
and $\mu$ exist, so $C(P)$ is consistent.

For the if direction, since $P$ does not have the UNION and OPT
operators, 
$\Gamma(P)$ is a singleton $\{S\}$.  Since $C(P)$ is consistent,
there exists a mapping
$\mu : S \to U$ satisfying every constraint in $C(P)$.
Let $G$ be the graph consisting of all triples
$(\mu(u),\mu(v),\mu(w))$ where $(u,v,w)$ is a triple pattern in
$P$.  It is straightforward to show by induction on the height of $Q$ that for
every subpattern $Q$ of $P$, we have $\mu|_{S'} \in \sem Q$,
where $\Gamma(Q)=\{S'\}$.  Hence $\mu \in \sem P$ and $P$ is
satisfiable.
\end{proof}

Note that $\Gamma(P)$ can ``blow up'' only because of possible
UNION and OPT operators, which are missing in an AF-pattern.
Hence, for an AF-pattern $P$, we can efficiently compute
$\Gamma(P)$ by a single bottom-up pass over $P$.  Morever, $C(P)$
is a conjunction of possibly negated equalities and constant
equalities.  It is well known that consistency of such
conjunctions can be decided in polynomial time
\cite{decisionprocedures}.  Hence, we conclude:

\begin{corollary} \label{afptime}
Satisfiability for AF-patterns can be checked in polynomial time.
\end{corollary}

\subsection{AF-reduction of well-designed patterns}

A well-designed pattern is defined as a union of union-free
well-designed patterns.  Since a union is satisfiable if and only
if one of its terms is, we will focus on union-free patterns in
what follows.  Formally, a union-free pattern $P$ is called
\emph{well-designed} \cite{perez_sparql_tods} if
\begin{enumerate}
\item
for every subpattern of $P$ of the form
$Q \, \FILTER \, C$, all variables mentioned in $C$ also occur in $Q$;
and
\item
for every subpattern $Q$ of $P$ of the form $Q_1 \OPT Q_2$, and every
$?x \in \var{Q_2}$, if $?x$ also occurs in $P$ outside of $Q$,
then $?x \in \var{Q_1}$.
\end{enumerate}

We associate to every union-free pattern $P$ an AF-pattern
$\rho(P)$ obtained by removing all applications of OPT and their
right operands; the left operand remains in place.  Formally, we
define the following:
\begin{itemize}
\item
If $P$ is a triple pattern, then $\rho(P)$ equals $P$.
\item
If $P$ is of the form $P_1 \AND P_2$, then $\rho(P) = \rho(P_1)
\AND \rho(P_2)$.
\item
If $P$ is of the form $P_1 \FILTER C$, then $\rho(P) = \rho(P_1)
\FILTER C$.
\item
If $P$ is of the form $P_1 \OPT P_2$, then $\rho(P) =
\rho(P_1)$.
\end{itemize}

For further use we note that $\Gamma(P)$ and
$\Gamma(\rho(P))$ are related in the following way.  The proof by
induction is straightforward.

\begin{lemma} \label{lemstuck}
Let $S \in \Gamma(P)$ and let $S' \in \Gamma(\rho(P))$.  Then $S'
\subseteq S$.
\end{lemma}

The announced result is now given by the following theorem.
The if direction of this theorem is already known from a result by
P\'erez et al.~\cite[Lemma~4.3]{perez_sparql_tods}.

\begin{theorem} \label{theorwell}
Let $P$ be a union-free well-designed pattern.  Then $P$ is
satisfiable if and only if $\rho(P)$ is.
\end{theorem}

Since $\rho(P)$ can be efficiently computed from $P$,
the above Theorem and Corollary~\ref{afptime} imply:

\begin{corollary} \label{ptimecor}
Satisfiability of union-free well-designed patterns can be tested
in polynomial time.
\end{corollary}

\subsection{Proof}

We prove the only-if direction of Theorem~\ref{theorwell}.
We begin by introducing two auxiliary
notations.
\begin{enumerate}
\item
For any pattern $P$ and subpattern $Q$ of $P$,
we denote by $\varp PQ$ the set
of variables from $\var Q$ that also occur in $P$ outside of $Q$.
\item
When $P$ is an AF-pattern with nonempty
$\Gamma(P)$, it is readily seen that $\Gamma(P)$ in that case
consists of a single scheme.  We denote the unique scheme in
$\Gamma(P)$ by $S(P)$.
\end{enumerate}

The following lemma connects the above two notations:
\begin{lemma} \label{lemrho}
Let $P$ be a union-free well-designed pattern, and let $Q$ be a
subpattern of $P$ such that $\Gamma(Q)$ is nonempty. Then
$\Gamma(\rho(Q))$ is nonempty as well, and
$\varp PQ \subseteq S(\rho(Q))$.
\end{lemma}
\begin{proof}
By induction on the height of $Q$.
If $Q$ is a triple pattern $(u, v, w)$, then 
we have $Q = \rho(Q)$ and
$\varp PQ \subseteq \var Q = \{u,v,w\}\cap V = S(Q) =
S(\rho(Q))$ as desired.

If $Q$ is of the form $Q_1 \AND Q_2$, then the definition of
$\Gamma(Q)$ immediately implies that $\Gamma(Q_1)$ and
$\Gamma(Q_2)$ must both be nonempty.  Since $\rho(Q)=\rho(Q_1)
\AND \rho(Q_2)$ we then obtain $S(\rho(Q)) = S(\rho(Q_1)) \cup
S(\rho(Q_2))$.  Any $?x \in \varp PQ$ belongs to $\varp P{Q_1}$
or $\varp P{Q_2}$; we assume the former case as the latter case
is analogous.  By induction, we then have $?x \in S(\rho(Q_1))
\subseteq S(\rho(Q))$ as desired.

If $Q$ is of the form $Q_1 \OPT Q_2$, then $\rho(Q)=\rho(Q_1)$.
Recall that $\Gamma(Q)=\Gamma(Q_1) \cup \Gamma(Q_1 \AND Q_2)$.
If $\Gamma(Q_1)$ is nonempty we obtain by induction that
$\Gamma(\rho(Q_1))=\Gamma(\rho(Q))$ is nonempty; if $\Gamma(Q_1
\AND Q_2)$ is nonempty we obtain $\Gamma(\rho(Q_1))$ nonempty as
in the case for AND\@.  So, $S(\rho(Q))$ exists and is equal to
$S(\rho(Q_1))$.  Now let $?x \in \varp PQ$.  If $?x \in \varp
P{Q_1}$ then $?x \in S(\rho(Q_1))$ by induction.  But if $?x \in
\varp P{Q_2}$, then also $?x \in \varp P{Q_1}$ since $P$ is
well-designed.  Hence we are done with this case.

Finally, let $Q$ be of the form $Q_1 \FILTER C$.  Since
$\Gamma(Q)$ is nonempty, $\Gamma(Q_1)$ is nonempty as well.  To
show that $\Gamma(\rho(Q))$ is nonempty we must show that
$S(\rho(Q_1)) \models C$.  Thereto, consider a variable $?x$
mentioned in $C$.  Since $P$ is well-designed, $?x \in \var
{Q_1}$ and thus $?x \in \varp P{Q_1}$.  By induction we obtain
$?x \in S(\rho(Q_1))$.  By Lemma~\ref{lemstuck}, then also $?x
\in S$ for every $S \in \Gamma(Q_1)$.  In other words, $S \not
\models \neg \bound(?x)$ for every $S \in \Gamma(Q_1)$.  This
rules out the possibility that $C$ is a negated bound-constraint,
since we are given that $\Gamma(Q)$ is nonempty.
On the other hand, this argument also shows that $S(\rho(Q_1))
\models C$ in the other cases, where $C$ is a bound-constraint or an
(constant) (non)equality, as desired.

It remains to show that $\varp PQ \subseteq S(\rho(Q)) =
S(\rho(Q_1))$.  Let $?x \in \varp PQ$.  If $?x \in \var {Q_1}$
the result follows by induction.  If $?x$ occurs in $C$ then,
because $P$ is well-designed, also $?x \in \var {Q_1}$ and thus
we are done.
\end{proof}

We mention in passing an interesting corollary of the
reasoning in the above proof, to
the effect that well-designedness rules out
any nontrivial use of negated bound-constraints:

\begin{corollary}
If $P$ is a union-free well-designed pattern and $Q$ is a
subpattern of $P$ of the form $Q_1
\FILTER {\neg \bound(?x)}$, then $\Gamma(Q)$ is empty, in
particular, $Q$ is unsatisfiable.
\end{corollary}

We are now ready to make the final step in the proof of
Theorem~\ref{theorwell}:

\begin{lemma} \label{lemlast}
Let $P$ be a union-free well-designed pattern.  If $\mu \in \sem P$
and $\Gamma(\rho(P))$ is nonempty, then $\mu|_{S(\rho(P))} \in \sem {\rho(P)}$.
\end{lemma}
\begin{proof}
By induction on the structure of $P$.  If $P$ is a triple pattern, then 
the claim is trivial.

So let $P$ be of the form $P_1 \AND P_2$.  Since $\Gamma(\rho(P))$
is nonempty and $\rho(P)=\rho(P_1) \AND \rho(P_2)$, also
$\Gamma(\rho(P_i))$ is nonempty for $i=1,2$.  Then by induction,
$\mu|_{S(\rho(P_i))} \in \sem {\rho(P_i)}$.  Since they are
restrictions of the same mapping $\mu$, we also have
$\mu|_{S(\rho(P_1))} \sim \mu|_{S(\rho(P_2))}$, so the mapping
$\mu|_{S(\rho(P_1))} \cup \mu|_{S(\rho(P_2))}$ belongs to $\sem
{\rho(P)}$.
Since $S(\rho(P)) = S(\rho(P_1)) \cup S(\rho(P_2))$, we
obtain $\mu|_{S(\rho(P))} \in \sem {\rho(P)}$ as desired.

If $P$ is of the form $P_1 \OPT P_2$, then we have
$\rho(P)=\rho(P_1)$, so we are given that $\Gamma(\rho(P_1))$ is
nonempty.  By induction, $\mu|_{S(\rho(P_1))} \in \sem
{\rho(P_1)} = \sem{\rho(P)}$ as desired.

Finally, if $P$ is of the form $P_1  \FILTER  C$ then by the
nonemptiness of $\Gamma(\rho(P))$ we know that
$S(\rho(P_1)) \models C$ 
and $S(\rho(P)) = S(\rho(P_1))$.  Hence, by induction,
$\mu|_{S(\rho(P_1))} \in \sem {\rho(P_1)}$.  It remains to show that
$\mu|_{S(\rho(P_1))} \models C$, but this
follows immediately because $\mu \models C$ and $S(\rho(P_1))
\models C$.
\end{proof}

With the above lemmas in hand,
the only-if direction of Theorem~\ref{theorwell} can now be 
argued as follows.  Since $P$ is satisfiable, $\Gamma(P)$ is
nonempty by Lemma~\ref{lemoif}.  By Lemma~\ref{lemrho} applied to
$Q=P$, also $\Gamma(\rho(P))$ is nonempty.  Since $P$ is
satisfiable, there exist $G$ and $\mu$ such that $\mu \in \sem
P$.  Now applying Lemma~\ref{lemlast} yields that $\sem
{\rho(P)}$ is nonempty.  We conclude that $\rho(P)$ is
satisfiable.

\section{Experimental evaluation} \label{secexp}

We want to evaluate experimentally the positive results
presented so far:
\begin{enumerate}
\item
Wrong literal reduction
(Proposition~\ref{propeasy});
\item
Satisfiability checking for the two fragments $\frageq$ and $\fragneq$ 
by computing $\Gamma(P)$ (Theorem~\ref{theordecidable});
\item
Satisifiability checking for well-designed patterns, by reduction
to AF-patterns (Proposition~\ref{propaf} and
Theorem~\ref{theorwell}).
\end{enumerate}

Our experiments follow up on those reported earlier by the third
author and Vansummeren \cite{pica_realsparql}.  As test datasets
of real-life SPARQL queries, we use logs of the SPARQL endpoint
for DBpedia, available at
\url{ftp://download.openlinksw.com/support/dbpedia/}.
This data source contains the ``query dumps'' from the year 2012,
divided into 14 logfiles.  Out of these we chose
the three logs 20120913, 20120929 and 20121031 to obtain a span
of roughly three months; we then took a
sample of $100\,000$ queries from each of them.
A typical query in the log has size between 75 and 125 (size
measured as number of nodes in the syntax tree).  About 10\% of
the queries in each log is not usable because they have syntax
errors or because they use features not covered by our analysis.

The implementation of the
       tests was done in Java~7 under Windows~7, on an Intel
       Core~2 Duo SU94000 processor (1.40GHz, 800MHz, 3MB) with
       3GB of memory (SDRAM DDR3 at 1067MHz).

Our tests measure the time needed to perform the analyses of
SPARQL queries presented above.  The timings are averaged over
all queries in a log, and each experiment is repeated five times
to smooth out accidental quirks of the operating system.
Although we give absolute timings, the main emphasis is on the
percentage of the time needed to analyse a query, with respect to
the time needed simply to read and parse that query.  If this
percentage is small this demonstrates efficient, linear time
complexity in practice.  It will turn out that this is indeed
achieved by our experiments, as shown in Table~\ref{tabwl}.

In the following subsections we discuss the results in more
detail.

\subsection{Wrong literal reduction}

Testing for and removing triple patterns with wrong literals in a
pattern $P$ is performed by the reduction $\lambda(P)$ defined in
the Appendix.  From the definition of
$\lambda(P)$ it is clear that it can be computed by a single
bottom-up traversal of $P$ and this is indeed borne out by our
experiments.  Table~\ref{tabwl} shows that on average, wrong-literal
reduction takes between 3 and 5\% of the time needed to read and
parse the input.

\begin{table}
\caption{Timings of experiments (averaged over five
repeats).  Times are in ms.  Baseline is
time to read and parse $1000\,000$ queries; WL stands for baseline
plus time for wrong-literal reduction.
$\Gamma(P)$ stands for WL
plus time for computing $\Gamma(P)$.
AF stands for baseline, plus testing
well-designedness, plus doing AF-reduction and testing
satisfiability (Proposition~\ref{propaf}).  The percentages show
the increases relative to the baseline.}
\label{tabwl}
\begin{center}
\begin{tabular}{rrrrrrrr}
logfile & baseline & WL & & $\Gamma(P)$ & & AF & \\
\hline
20120913 & $39\,422$ & $41\,254$ & $5\%$ & $44\,395$ & $8\%$ &
    $48\,329$ & $10\%$ \\
20120929 & $34\,281$ & $35\,868$ & $5\%$ & $38\,102$ & $7\%$ &
    $41\,087$ & $9\%$ \\
20121031 & $32\,286$ & $33\,186$ & $3\%$ & $34\,419$ & $4\%$ &
     $36\,993$ & $8\%$
\end{tabular}
\end{center}
\end{table}

Interestingly, some real-life queries with literals in the wrong
position were indeed found; one example is the following:

\begin{small}
\begin{verbatim}
SELECT DISTINCT *
WHERE { 49  dbpedia-owl:wikiPageRedirects  ?redirectLink .}
\end{verbatim}
\end{small}

\subsection{Computing $\Gamma(P)$}

In Section~\ref{secsat} we have seen that satisfiability for the
decidable fragments can be tested by computing
$\Gamma(P)$, but that the problem is NP-complete.  Intuitively,
the problem is intractable because $\Gamma(P)$ may be of size exponential
in the size of $P$.  This actually occurs in real
life; a common SPARQL query pattern is to use many nested
OPTIONAL operators to gather additional information that is not
strictly required by the query but may or may not be present.
We found in our experiments queries with up to 50 nested
OPT operators, which naively would lead to a $\Gamma(P)$ of size
$2^{50}$.  A shortened example of such a query is shown in
Figure~\ref{figexp}.

\begin{figure}
\begin{footnotesize}
\begin{verbatim}
SELECT DISTINCT *
WHERE {
?s a <http://dbpedia.org/ontology/EducationalInstitution>,
<http://dbpedia.org/ontology/University> .
?s <http://dbpedia.org/ontology/country> <http://dbpedia.org/resource/Brazil> .
OPTIONAL {?s  <http://dbpedia.org/ontology/affiliation> ?ontology_affiliation .}
OPTIONAL {?s  <http://dbpedia.org/ontology/abstract> ?ontology_abstract .}
OPTIONAL {?s  <http://dbpedia.org/ontology/campus> ?ontology_campus .}
OPTIONAL {?s  <http://dbpedia.org/ontology/chairman> ?ontology_chairman .}
OPTIONAL {?s  <http://dbpedia.org/ontology/city> ?ontology_city .}
OPTIONAL {?s  <http://dbpedia.org/ontology/country> ?ontology_country .}
OPTIONAL {?s  <http://dbpedia.org/ontology/dean> ?ontology_dean .}
OPTIONAL {?s  <http://dbpedia.org/ontology/endowment> ?ontology_endowment .}
OPTIONAL {?s  <http://dbpedia.org/ontology/facultySize> ?ontology_facultySize .}
OPTIONAL {?s  <http://dbpedia.org/ontology/formerName> ?ontology_formerName .}
OPTIONAL {?s  <http://dbpedia.org/ontology/head> ?ontology_head .}
OPTIONAL {?s  <http://dbpedia.org/ontology/mascot> ?ontology_mascot .}
OPTIONAL {?s  <http://dbpedia.org/ontology/motto> ?ontology_motto .}
OPTIONAL {?s  <http://dbpedia.org/ontology/president> ?ontology_president .}
OPTIONAL {?s  <http://dbpedia.org/ontology/principal> ?ontology_principal .}
OPTIONAL {?s  <http://dbpedia.org/ontology/province> ?ontology_province .}
OPTIONAL {?s  <http://dbpedia.org/ontology/rector> ?ontology_rector .}
OPTIONAL {?s  <http://dbpedia.org/ontology/sport> ?ontology_sport .}
OPTIONAL {?s  <http://dbpedia.org/ontology/state> ?ontology_state .}
OPTIONAL {?s  <http://dbpedia.org/property/acronym> ?property_acronym .}
OPTIONAL {?s  <http://dbpedia.org/property/address> ?property_address .}
OPTIONAL {?s  <http://www.w3.org/2003/01/geo/wgs84_pos#lat> ?property_lat .}
OPTIONAL {?s  <http://www.w3.org/2003/01/geo/wgs84_pos#long> ?property_long .}
OPTIONAL {?s  <http://dbpedia.org/property/established> ?property_established .}
OPTIONAL {?s  <http://dbpedia.org/ontology/logo> ?ontology_logo .}
OPTIONAL {?s  <http://dbpedia.org/property/website> ?property_website .}
OPTIONAL {?s  <http://dbpedia.org/property/location> ?property_location .}
FILTER ( langMatches(lang(?ontology_abstract), "es") ||
langMatches(lang(?ontology_abstract), "en") )
FILTER ( langMatches(lang(?ontology_motto), "es") ||
langMatches(lang(?ontology_motto), "en") )
}
\end{verbatim}
\end{footnotesize}
\caption{A real-life query with many nested OPTIONAL operators,
retrieving as much information as possible about universities in
Brazil.}
\label{figexp}
\end{figure}

In practice, however, the blowup of $\Gamma(P)$ can be avoided as
follows.  Recall that Theorem~\ref{theordecidable} states that
$P$ is satisfiable if and only if $\Gamma(P)$ is nonempty.  The
elements of $\Gamma(P)$ are sets of variables.  Looking at the
definition of $\Gamma(P)$, a set may be removed from $\Gamma(P)$
only by the application of a FILTER\@.  Hence, only variables that are
mentioned in FILTER conditions can influence the emptiness
of $\Gamma(P)$; other variables can be ignored.  For example, in the query in
Figure~\ref{figexp}, only two variables appear in a filter,
namely {\small \texttt{?ontology\_abstract}} and
{\small \texttt{?ontology\_motto}}, so that the
maximal size of $\Gamma(P)$ is reduced to $2^2$.

In our experiments, it turns out that typically few variables are
involved in filter conditions.  Hence, the above strategy works
well in practice.

Another practical issue is that, in this paper, we have only
considered filter conditions that are bound checks, equalities,
and constant-equalities, possibly negated.  In practice,
filter conditions typically apply built-in SPARQL predicates such
as the predicate \texttt{langMatches} in Figure~\ref{figexp}.
For the experimental purpose of testing the practicality of computing
$\Gamma(P)$, however, such predicates can simply be treated as bound
checks.  In this way we can apply our experiments to 70\% of the queries
in the testfiles.

With the above practical adaptations, our experiments show that
computing $\Gamma(P)$ is efficient: Table~\ref{tabwl} shows
that it requires, on average, between 4 and 8\% of the time
needed to read and parse the input, and these timings even
include the wrong-literal reduction.

\subsection{Satisfiability testing for well-designed patterns}

In Section~\ref{secwell} we have seen that testing satisfiability
of a well-designed pattern can be done by testing satisfiability of
the AF-reduction (Theorem~\ref{theorwell}).  The latter can be
done by testing nonemptiness of $\Gamma(P)$ and testing
consistency of the filter conditions (Proposition~\ref{propaf}).

Computing the AF-reduction can be done by a simple bottom-up
traversal of the pattern.  Moreover, for an AF-pattern $P$,
computing $\Gamma(P)$ poses no problems since it is either empty
or a singleton.  As far as testing consistency of filter
conditions is concerned, our experiments yield a rather baffling
observation: almost all well-designed patterns in the test sets
have no filters at all.  We cannot explain this phenomenon, but
it implies that we have not been able to test the performance of
the consistency checks on real-life SPARQL queries.

Anyhow, Table~\ref{tabwl} shows that doing the entire
analysis of wrong-literal reduction, testing well-designedness,
AF-reduction, computing $\Gamma(P)$, and consistency checking (in
the few cases where the latter was necessary), incurs at most a 10\%
increase relative to reading and parsing the input.

\subsection{Scalability}

The experiments described above were run on sets of $100\,000$
queries each.  We also did a modest scaling experiment where we varied the
number of queries from $5\,000$ to $200\,000$.  Table~\ref{tabscale}
shows that the performance scales linearly.

\begin{table}
\caption{Scalability experiment (times in ms).
Timings clearly scale linearly for increasing input size.}
\label{tabscale}
\begin{center}
\begin{tabular}{rrrrrrl}
input size & $200\,000$ & $100\,000$ & $50\,000$ & $10\,000$ &
$5\,000$ & Pearson coeficient \\
\hline
baseline & $74\,168$ & $39\,422$ &$21\,315$ &$3\,596$ & $1\,851$ &$0.999924005$ \\
WL & $77\,800$&$41\,253$&$21\,876$&$3\,762$&$1\,942$&$0.999989454$ \\
$\Gamma(P)$& $81\,730$&$44\,395$&$23\,552$&$4\,016$&$2\,036$&$0.999900948$
\\
AF& $91\,470$&$48\,329$&$26\,023$&$4\,463$&$2\,254$&$0.999044542$
\end{tabular}
\end{center}
\end{table}

\section{Extension to SPARQL 1.1} \label{sec1.1}

As already mentioned in the Introduction, SPARQL~1.0 has been
extended to SPARQL~1.1 with a number of new operators for
building patterns.  The main new features are property paths;
grouping and aggregates; BIND; VALUES; MINUS;
EXISTS and NOT EXISTS-subqueries; and
SELECT\@.  A complete analysis of SPARQL~1.1 goes
beyond the scope of the present paper.  Nevertheless, in this
section, we briefly discuss how our results may be extended to
this new setting.

Property paths provide a form of regular path querying over
graphs.  This aspect of graph querying has already been
extensively investigated, including questions of satisfiability
and other kinds of static analysis such as query containment
\cite{vrgoc_containment,krrv_sparqlpp}.  Therefore we do not discuss
property paths any further here.

The SPARQL~1.1 features that we discuss can be grouped in two
categories: those that cause undecidability, and those that are
harmless as far as satisfiability is concerned.  We begin with
the harmless category.

\subsection{SELECT operator and EXISTS-subqueries}

SPARQL~1.1 allows patterns of the form $\SELECT_S P$, where $S$
is a finite set of variables and $P$ is a pattern.  The semantics
is that of projection: solution mappings are restricted to the
variables listed in $S$.  Formally, we define $$ \sem {\SELECT_S
P} = \{\mu|_{S \cap \dom \mu} \mid \mu \in \sem P\}. $$  

This feature in itself does not influence the satisfiability of
patterns.  Indeed, patterns extended with SELECT
operators can be reduced to patterns without said operators.  The
reduction amounts simply to rename the variables that are
projected out by fresh variables that are not used anywhere else
in the pattern; then the SELECT operators themselves can be
removed.  The resulting, SELECT-free, pattern is equivalent to
the original one if we omit the fresly introduced variables from
the solution mappings in the final result.  In particular,
the two patterns are equisatisfiable.

\begin{example}
Rather than giving the formal definition of SELECT-reduction and
formally stating and proving the equivalence, we give an example.
Consider the pattern $P$:
$$ (c,p,?x) \OPT ( (?x,p,?y) \AND
\SELECT_{?y} (?y,q,?z) \AND 
\SELECT_{?y} (?y,r,?z)
) $$
Renaming projected-out variables by fresh variables and omitting
the SELECT operators yields the
following pattern $P'$:
$$ (c,p,?x) \OPT ( (?x,p,?y) \AND
(?y,q,?z_1) \AND (?y,r,?z_2)
) $$
Pattern $P'$ is equivalent to $P$ in the sense that for any graph
$G$, we have $\sem P = \{\hat \mu \mid \mu \in \sem {P'}\}$,
where $\hat \mu$ denotes the mapping obtained from $\mu$ by
omitting the values for $?z_1$ and $?z_2$ (if at all present in
$\dom \mu$).
\qed
\end{example}

Now that we know how to handle SELECT operators, we can also
handle EXISTS-subqueries.  Indeed, a pattern $P \, \FILTER \, {\EXISTS
(Q)}$ (with the obvious SQL-like semantics) is equivalent to
$\SELECT_{\var P}(P \AND Q)$.

\subsection{Features leading to undecidability}

In Section~\ref{secund} we have seen that as soon as one can
express the union, composition and difference of binary
relations, the satisfiability problem becomes undecidable.  Since
union and composition are readily expressed in basic SPARQL
($\UNION$ and $\AND$), the key lies in the expressibility of the
difference operator.  In this subsection we will see that various
new features of SPARQL~1.1 indeed allow expressing difference.

\paragraph{MINUS operator and NOT EXISTS subqueries}
Any of these two features can quite obviously be used
to express difference, so we do not dwell on them any further.

\paragraph{Grouping and aggregates}  A known trick for expressing
difference using grouping and counting \cite{sqlforsmarties} can be
emulated in the extension of SPARQL~1.0 with grouping.  We
illustrate the technique with an example.

Consider the query $(?x,p,?y) \MINUS (?x,q,?y)$ asking for all
pairs $(a,b)$ such that $(a,p,b)$ holds but $(a,q,b)$ does not.
We can express this query (with the obvious SQL-like semantics) as follows:
\begin{tabbing}
$\SELECT_{?x,?y} \bigl ( 
(?x,p,?y) \OPT ((?x,q,?y) \AND (?xx,p,?yy)) \bigr )$ \\
$\mathrm{{GROUP\ BY}}\ {?x,?y}$ \\
$\mathrm{HAVING} \ {\mathrm{count}(?xx) = 0}$
\end{tabbing}
Note that this technique of looking for the $(?x,?y)$ groups with
a zero count for $?xx$ is very similar to the technique used to
express difference using a negated bound constraint (seen in the
proof of Lemma~\ref{reductionlemma}).

\paragraph{BIND and VALUES}

We have seen in Section~\ref{secfrageqc} that allowing constant
equalities in filter constraints allows us to emulate the
difference operator.  Two mechanisms introduced in SPARQL~1.1, BIND and
VALUES, allow the introduction of constants in solution mappings.
Together with equality constraints this allows us to express
constant equalities, and hence, difference.

Specifically, using VALUES, we can express $P \, \FILTER \, {?x=c}$
as $$ \SELECT_{\var P}(P \AND \VALUES_{?x}(c)). $$  Using BIND, it
can be expressed as $$ \SELECT_{\var P}((P \BIND_{?x'}(c)) \FILTER
{?x = {?x'}}) $$ where $?x'$ is a fresh variable.  Note the use of
SELECT, which, however, does not influence satisfiability as
discussed above.  We conclude that SPARQL($=$) extended with
BIND, or SPARQL($=$) extended with VALUES, have an undecidable
satisfiability problem.

\section{Conclusion} \label{seconcl}

The results of this paper may be summarized by saying that, as
long as the kinds of constraints allowed in filter conditions
cannot be combined to yield inconsistent sets of constraints,
satisfiability for SPARQL patterns is decidable; otherwise, the
problem is undecidable.  Moreover, for well-designed patterns,
satisfiability is decidable as well.  All our positive results
yield straightforward bottom-up syntactic checks that can be 
implemented efficiently in practice.

We thus have attempted to paint a rather complete picture of the
satisfiability problem for SPARQL~1.0\@.  Of course, satisfiability
is only the most basic automated reasoning task. One may now
move on to more complex tasks such as equivalence, implication,
containment, or query answering over ontologies.
Indeed, investigations along this line for limited
fragments of SPARQL are already happening
\cite{pp_sparqlcontainment,fransen_sparqlcontain,kg_sparqlontology,ox_cgiortp}
and we hope that our work may
serve to provide some additional grounding to these investigations.

We also note that in query optimization it is standard to check
for satisfiability of subexpressions, to avoid executing useless
code.  Some specific works on SPARQL query optimization
\cite{sequeda_ultrawrap,groppe_sparqlosers} do mention that
inconsistent constraints can cause unsatisfiability, but they
have not provided sound and complete characterizations of
satisfiability, like we have offered in this paper.  Thus, our
results will be useful in this direction as well.

\section*{Acknowledgment}

We thank the anonymous referees for their critical comments on a
previous version of this paper, which encouraged us to
significantly improve the paper.

\bibliographystyle{alpha}
\bibliography{database}

\begin{thebibliography}{CGMSH12}

\bibitem[ACP12]{chili_yotta}
M.~Arenas, S.~Conca, and J.~P\'erez.
\newblock Counting beyond a {Y}ottabyte, or how {SPARQL} 1.1 property paths
  will prevent adoption of the standard.
\newblock In A.~Mille et~al., editors, {\em Proceedings 21st World Wide Web
  Conference}, pages 629--638. ACM, 2012.

\bibitem[AG08]{ag_expsparql}
R.~Angles and C.~Gutierrez.
\newblock The expressive power of {SPARQL}.
\newblock In A.~Sheth, S.~Staab, et~al., editors, {\em Proceedings 7th
  International Semantic Web Conference}, volume 5318 of {\em Lecture Notes in
  Computer Science}, pages 114--129. Springer, 2008.

\bibitem[AGN97]{andreka_memoir}
H.~Andr\'eka, S.~Givant, and I.~N\'emeti.
\newblock {\em Decision problems for equational theories of relational
  algebras}, volume 126 of {\em Memoirs}.
\newblock AMS, 1997.

\bibitem[AHV95]{ahv_book}
S.~Abiteboul, R.~Hull, and V.~Vianu.
\newblock {\em Foundations of Databases}.
\newblock Addison-Wesley, 1995.

\bibitem[AP11]{chile_sparql_pods}
M.~Arenas and J.~P\'erez.
\newblock Querying semantic web data with {SPARQL}.
\newblock In {\em Proceedings 30st ACM Symposium on Principles of Databases},
  pages 305--316. ACM, 2011.

\bibitem[APG09]{semanticsparql}
M.~Arenas, J.~P\'erez, and C.~Gutierrez.
\newblock On the semantics of {SPARQL}.
\newblock In R.~De~Virgilio, F.~Giunchiglia, and L.~Tanca, editors, {\em
  Semantic Web Information Management---A Model-Based Perspective}, pages
  281--307. Springer, 2009.

\bibitem[Cel05]{sqlforsmarties}
J.~Celko.
\newblock {\em {SQL} for Smarties: Advanced {SQL} Programming}.
\newblock Elsevier, third edition, 2005.

\bibitem[CGMSH12]{ox_cgiortp}
B.~Cuenca~Grau, B.~Motik, G.~Stoilos, and I.~Horrocks.
\newblock Completeness guarantees for incomplete ontology reasoners: Theory and
  practice.
\newblock {\em Journal of Artificial Intelligence Research}, 43:419--476, 2012.

\bibitem[GGK09]{groppe_sparqlosers}
J.~Groppe, S.~Groppe, and J.~Kolbaum.
\newblock Optimization of {SPARQL} by using {coreSPARQL}.
\newblock In J.~Cordeiro and J.~Filipe, editors, {\em Proceedings 11th
  International Conference on Enterprise Information Systems}, pages 107--112,
  2009.

\bibitem[KG13]{kg_sparqlontology}
I.~Kollia and B.~Glimm.
\newblock Optimizing {SPARQL} query answering over {OWL} ontologies.
\newblock {\em Journal of Artificial Intelligence Research}, 48:253--303, 2013.

\bibitem[KRRV15]{krrv_sparqlpp}
E.V. Kostylev, J.L. Reutter, M.~Romero, and D.~Vrgo\v{c}.
\newblock {SPARQL} with property paths.
\newblock In M.~Arenas, O.~Corcho, E.~Simperl, M.~Strohmaier, et~al., editors,
  {\em Proceedings 14th International Semantic Web Conference}, volume 9366 of
  {\em Lecture Notes in Computer Science}, pages 3--18. Springer, 2015.

\bibitem[KRV14]{vrgoc_containment}
E.V. Kostylev, J.L. Reutter, and D.~Vrgo\v{c}.
\newblock Containment of data graph queries.
\newblock In {\em Proceedings 17th International Conference on Database
  Theory}. ACM, 2014.

\bibitem[KS08]{decisionprocedures}
D.~Kroening and O.~Strichman.
\newblock {\em Decision Procedures}.
\newblock Springer, 2008.

\bibitem[LPPS13]{pp_sparqlcontainment}
A.~Letelier, J.~P\'erez, R.~Pichler, and S.~Skritek.
\newblock Static analysis and optimization of semantic web queries.
\newblock {\em ACM Transactions on Database Systems}, 38(4):article 25, 2013.

\bibitem[PAG09]{perez_sparql_tods}
J.~P\'erez, M.~Arenas, and C.~Gutierrez.
\newblock Semantics and complexity of {SPARQL}.
\newblock {\em ACM Transactions on Database Systems}, 34(3):article 16, 2009.

\bibitem[Pol07]{polleres_sparqldatalog}
A.~Polleres.
\newblock From {SPARQL} to rules (and back).
\newblock In C.L. Williamson, M.E. Zurko, et~al., editors, {\em Proceedings
  16th World Wide Web Conference}, pages 787--796. ACM, 2007.

\bibitem[PV11]{pica_realsparql}
F.~Picalausa and S.~Vansummeren.
\newblock What are real {SPARQL} queries like?
\newblock In R.~De~Virgilio, F.~Giunchiglia, and L.~Tanca, editors, {\em
  Proceedings International Workshop on Semantic Web Information Management},
  page article 7. ACM Press, 2011.

\bibitem[RDF04]{RDFprimer}
{RDF} primer.
\newblock W3C Recommendation, February 2004.

\bibitem[SM13]{sequeda_ultrawrap}
J.F. Sequeda and D.P. Miranker.
\newblock Ultrawrap: {SPARQL} execution on relational data.
\newblock {\em Web Semantics}, 22:19--39, 2013.

\bibitem[SML10]{schmidt_sparqloptim}
M.~Schmidt, M.~Meier, and G.~Lausen.
\newblock Foundations of {SPARQL} query optimization.
\newblock In {\em Proceedings 13th International Conference on Database
  Theory}, pages 4--33. ACM, 2010.

\bibitem[SPA08]{sparql}
{SPARQL} query language for {RDF}.
\newblock W3C Recommendation, January 2008.

\bibitem[SPA13]{sparql1.1}
{SPARQL} 1.1 query language.
\newblock W3C Recommendation, March 2013.

\bibitem[TVdBZ14]{tony_da_arxiv}
T.~Tan, J.~Van~den Bussche, and X.~Zhang.
\newblock Undecidability of satisfiability in the algebra of finite binary
  relations with union, composition, and difference.
\newblock arXiv:1406.0349, 2014.

\bibitem[WEGL12]{fransen_sparqlcontain}
M.~Wudage, J.~Euzenat, P.~Genev\`es, and N.~Laya\"{\i}da.
\newblock {SPARQL} query containment under {SHI} axioms.
\newblock In {\em Proceedings 26th AAAI Conference}, pages 10--16, 2012.

\end{thebibliography}

\appendix

\section*{Appendix}

Literals in the wrong place in triple patterns are easily dealt
with in the following manner.
We define the \emph{wrong-literal reduction} of a pattern $P$, denoted by
$\lambda(P)$, as a set that is either empty or is a singleton
containing a single pattern $P'$:
\begin{itemize}
\item
If $P$ is a triple pattern $(u,v,w)$ and $u$ is a literal, then
$\lambda(P):=\emptyset$; else $\lambda(P):=\{P\}$.
\item
$\lambda(P_1 \UNION P_2) := \lambda(P_1) \cup \lambda(P_2)$ if
$\lambda(P_1)$ or $\lambda(P_2)$ is empty;
\item
$\lambda(P_1 \UNION P_2) := \{P_1' \UNION P_2' \mid P_1' \in
\lambda(P_1)$ and $P_2' \in \lambda(P_2)\}$ otherwise.
\item
$\lambda(P_1 \AND P_2) := \{P_1' \AND P_2' \mid P_1' \in
\lambda(P_1)$ and $P_2' \in \lambda(P_2)\}$.
\item
$\lambda(P_1 \OPT P_2) := \emptyset$ if $\lambda(P_1)$ is empty;
\item
$\lambda(P_1 \OPT P_2) := \lambda(P_1)$ if $\lambda(P_2)$ is empty but
$\lambda(P_1)$ is nonempty;
\item
$\lambda(P_1 \OPT P_2) := \{P_1' \OPT P_2' \mid P_1' \in
\lambda(P_1)$ and $P_2' \in \lambda(P_2)\}$ otherwise.
\item
$\lambda(P_1 \FILTER C) := \{P_1' \FILTER C \mid P_1' \in
\lambda(P_1)\}$.
\end{itemize}

Note that the wrong-literal reduction never has a literal in the
subject position of a triple pattern.  The next proposition shows
that, as far as satisfiability checking is concerned, we may
always perform the wrong-literal reduction.

\begin{proposition} \label{propeasy}
Let $P$ be a pattern. If $\lambda(P)$ is empty then $P$ is unsatisfiable; if
$\lambda(P)=\{P'\}$
then $P$ and $P'$ are equivalent, i.e., $\sem P  = \sem {P'}$
for every RDF graph $G$.  Moreover, if $\lambda(P)=\{P'\}$ then $P'$
does not contain any triple pattern $(u,v,w)$ where $u$ is a
literal.
\end{proposition}
\begin{proof}
Assume $P$ is a triple pattern $(u,v,w)$ and $u$ is a literal,
so that $\lambda(P)=\emptyset$.  Since $u$ is a constant, 
$\mu(u)$ equals the literal $u$ for every solution mapping $\mu$.
Since no triple in an RDF graph can have a literal in its
first position, $\sem P$ is empty for every RDF graph $G$, i.e.,
$P$ is unsatisfiable.
If $u$ is not a literal, $\lambda(P)=\{P\}$ and the claims of the
Proposition are trivial.

If $P$ is of the form $P_1\UNION P_2$, or $P_1 \AND P_2$, or $P_1
\FILTER C$, the claims of the Proposition
follow straightforwardly by induction.

If $P$ is of the form $P_1\, \OPT\, P_2$, there are three cases
to consider.
\begin{itemize}
\item If $\lambda(P_1)$ is empty then so is $\lambda(P)$. In this
case, by induction, $P_1$ is unsatisfiable, whence so is $P$.
\item
If $\lambda(P_1) = \{P_1'\}$ is nonempty but $\lambda(P_2)$ is empty,
then $\lambda(P)=\{P_1'\}$.  By induction,
$P_2$ is unsatisfiable.  Hence, $P$ is equivalent to
$P_1$, which in turn is equivalent to $P_1'$ by induction.
That $P_1'$ does not contain any triple pattern with a literal in
first position again follows by induction.
\item
If $\lambda(P_1)=\{P_1'\}$ and $\lambda(P_2)=\{P_2'\}$ are both
nonempty, then $\lambda(P)=P_1' \OPT P_2'$. By induction,
$P_1$ is equivalent to $P_1'$ and so is $P_2$ to $P_2'$.  Hence,
$P$ is equivalent to $P_1' \OPT P_2'$ as desired.  By induction,
neither $P_1'$ nor $P_2'$ contain any triple pattern with a
literal in first position, so neither does $P_1' \OPT P_2'$.
\end{itemize}
\end{proof}
\end{document}